\documentclass{article}

\usepackage{arxiv}

\usepackage[utf8]{inputenc}
\usepackage[T1]{fontenc}
\usepackage{hyperref}
\usepackage{url}
\usepackage{booktabs}
\usepackage{amsfonts}
\usepackage{amsmath}
\usepackage{amssymb}
\usepackage{amsthm}
\usepackage{nicefrac}
\usepackage{microtype}
\usepackage{graphicx}
\usepackage{doi}
\usepackage{algorithm}
\usepackage{algpseudocode}

\newtheorem{theorem}{Theorem}

\newtheorem{lemma}{Lemma}
\newtheorem{proposition}{Proposition}
\newtheorem{corollary}{Corollary}
\newtheorem{remark}{Remark}

\title{Parent-Hash DAG:\\
A Cost Analysis of Constant-Time Append\\
for On-Chain Registries}

\author{
  \href{https://orcid.org/0009-0001-6031-4066}{\includegraphics[scale=0.06]{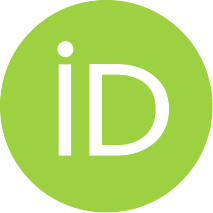}\hspace{1mm}Ian C.\ Moore, PhD}\\
  Principal, DeFiMind\\
  \texttt{imoore@defimind.ai}
  \and
  \href{https://orcid.org/0009-0009-3400-3545}{\includegraphics[scale=0.06]{orcid.pdf}\hspace{1mm}Fernando Paredes Garc\'ia}\\
  Researcher, DeFiMind \\
  \texttt{fparedes@defimind.ai}
}

\renewcommand{\shorttitle}{The Parent-Hash DAG}

\hypersetup{
pdftitle={Parent-Hash DAG: A Cost Analysis of Constant-Time Append for On-Chain Registries},
pdfsubject={cs.DC},
pdfauthor={Ian C. Moore, Fernando Paredes Garcia},
pdfkeywords={blockchain, append-only registry, provenance trees, Merkle tree, parent-hash DAG, gas cost, stochastic analysis, O(1) complexity, Ethereum, provenance},
}

\begin{document}
\maketitle

\begin{abstract}
Provenance trees are append-only directed acyclic graphs of artifact registrations anchored on a public blockchain, recently introduced as the data substrate of operator-gated provenance infrastructure. Their defining data-structural pattern is a parent-hash directed acyclic graph (PHDAG), in which each append performs a constant number of storage writes to previously-untouched slots. This pattern has not previously been isolated as a standalone primitive, formally bounded with explicit constants, or benchmarked against the standard alternative, the incremental Merkle tree (IMT). We formalize PHDAG append as O(1) in gas cost, independent of registry size and tree depth, and develop a stochastic cost model for IMT in which per-insert cost is a random variable over the leaf index, deriving closed-form expressions for its mean and variance. We validate both analyses empirically on Base Sepolia across tree depths 1 to 25. PHDAG is observed to be depth-invariant at 76,276 gas (standard deviation about 6 gas), while IMT cost grows linearly with depth. The crossover below which IMT is cheaper falls far beneath the depths of every production registry surveyed. We further establish trustless registry reconstruction from public event logs in linear time with no off-chain dependency.
\end{abstract}


\section{Introduction}
\label{section:intro}

Provenance trees (PTs) are append-only directed acyclic graphs of
artifact registrations anchored on a public blockchain, recently
introduced as the data substrate of operator-gated provenance
infrastructure for digital artifacts~\cite{Moore26Trustless}. Each node
in a PT carries a globally unique identifier, a parent reference
binding it to its lineage, and a manifest commitment; the resulting
structure is tamper-evident, publicly verifiable, and queryable by
both humans and AI agents. This queryability is not incidental:
a PT functions as a coordination medium in the sense of Paredes
Garc\'ia~\cite{Paredes26}, with downstream agents (verifiers,
derivative-work registries, AI provenance tools) acting on the
traces left by prior registrations rather than via direct message
exchange. PTs are advanced
in~\cite{Moore26Trustless} together with an assertion that the
underlying append operation is $O(1)$, invariant to both registry
size and tree depth, contrasted with the $\Theta(\log N)$ of
Merkle-tree registries. The assertion is stated at a system level,
in the context of a specific cryptographic commitment scheme for
operator-gated attribution. It does not isolate the underlying
data-structural pattern as a standalone primitive, does not
characterize the distribution of per-insert cost, does not state the
constants underlying the complexity classes, and does not establish
the empirical crossover depth at which the pattern becomes cheaper
than its standard alternative, the incremental Merkle tree.

The present paper closes that gap. We isolate the data-structural
nucleus of the PT construction, the parent-hash directed acyclic
graph (PHDAG) primitive, in which each append performs a constant
number of storage writes to previously-untouched slots without
reference to tree geometry, and place it in direct comparison with the IMT primitive
that dominates production zero-knowledge protocols and rollup state
commitments. The IMT, introduced in production settings by Tornado
Cash~\cite{TornadoCash19} and adopted subsequently by Semaphore and
the major zero-knowledge rollup state commitments, maintains a
fixed-depth binary tree in which each append updates a
logarithmic-length frontier. IMT serves in this paper as the
comparator against which the PHDAG primitive's properties are made
legible: depth-dependent cost versus depth-invariant cost, growing
variance versus near-zero variance, succinct inclusion proofs versus
event-log reconstructibility.

The IMT occupies a singular position in the deployed Ethereum
ecosystem. It is the data structure underlying Tornado Cash's
anonymity sets~\cite{TornadoCash19}, Semaphore's zero-knowledge
identity commitments~\cite{Semaphore20}, and the state
commitments of the major zero-knowledge rollups: zkSync~\cite{ZkSync22},
Scroll~\cite{Scroll23}, and Linea~\cite{Linea23}. The pattern's appeal
in these settings is a single property: the $O(\log N)$ Merkle
inclusion proof, which permits a contract or off-chain verifier to
confirm membership of a leaf without reference to the remainder of
the tree. This property is essential to zero-knowledge applications,
where in-circuit verification of membership must complete within a
single transaction's gas budget. The corresponding append cost,
growing logarithmically in tree capacity, is accepted as the price
of admission. For applications where succinct membership proofs are
not the dominant query, the cost trade is less obviously favorable,
and to our knowledge has not been examined.

This paper makes four contributions:

\begin{enumerate}
\item \textbf{Formal complexity analysis.} We formalize the $O(1)$
      append complexity of PHDAG with explicit bounds on the number
      of EVM operations per transaction, and contrast it with the
      $O(d)$ complexity of IMT. We further characterize the
      off-chain reconstruction complexity as $O(|V|)$
      (Section~\ref{section:complexity}).

\item \textbf{Stochastic cost model.} We model per-insert gas cost as
      a random variable over the uniformly distributed leaf index.
      For IMT at depth $d$ we derive closed-form expressions for the
      mean and variance of per-insert cost, obtaining
      $\mathbb{E}[g] = c_0 + \tfrac{d}{2}(c_L + c_R)$ and
      $\mathrm{Var}[g] = (c_L - c_R)^2 d / 4$, where $c_L$ and $c_R$
      are the per-level costs of a write level and a read level under
      the EVM storage-access gas schedule
      (Section~\ref{section:stochastic}).

\item \textbf{Empirical validation.} We deploy independent IMT and
      PHDAG contracts on Base Sepolia across depths $1$--$25$ and
      measure per-insert gas cost for proportional workloads, totaling
      over $500$ transactions (Section~\ref{section:results}).

\item \textbf{Crossover identification.} We empirically locate the
      depth at which PHDAG becomes cheaper than IMT, and situate this
      crossover relative to production-deployed systems. The crossover
      falls below every surveyed production registry, placing all
      such systems in the regime where PHDAG is strictly cheaper
      and exhibits negligible per-insert variance
      (Section~\ref{section:discussion}).
\end{enumerate}

The paper is structured as follows. Section~\ref{section:background}
reviews the two primitives and establishes notation.
Section~\ref{section:complexity} presents the formal complexity
analysis. Section~\ref{section:stochastic} develops the stochastic
cost model. Section~\ref{section:methodology} describes the
experimental setup. Section~\ref{section:results} presents the
empirical results. Section~\ref{section:discussion} discusses the
regime of applicability of each primitive and the implications for
standardization. Section~\ref{section:related} surveys related work
and Section~\ref{section:conclusion} concludes.

\section{Background and Notation}
\label{section:background}

\subsection{The EVM Gas Schedule}
\label{section:gas-schedule}

We work under the EVM gas schedule established by EIP-2929 and
EIP-3529~\cite{EIP2929,EIP3529}, which introduced the
\emph{cold-versus-warm access model}
for storage operations. Under this model, the EVM maintains a
per-transaction access list of storage slots that have been touched
during the transaction's execution. The first access to a slot is
\emph{cold} and is billed at a higher rate; every subsequent access
to the same slot within the same transaction is \emph{warm} and is
billed at a lower rate. The access list is reset at the start of
every transaction.

The schedule applies to both reads (\texttt{SLOAD}) and writes
(\texttt{SSTORE}), with distinct cold and warm prices for each. The
zero-versus-nonzero value-transition cases for \texttt{SSTORE} follow
the net-metering definitions of EIP-2200~\cite{EIP2200}, and the
underlying opcode and receipt semantics follow the Ethereum Yellow
Paper~\cite{YellowPaper}:
\begin{itemize}
  \item \textbf{Cold SSTORE}: $\approx 22{,}100$ gas when writing to
        a storage slot not yet accessed in the current transaction.
  \item \textbf{Warm SSTORE}: $\approx 5{,}000$ gas when writing to
        a slot already accessed in the current transaction.
  \item \textbf{Cold SLOAD}: $\approx 2{,}100$ gas on first access
        in a transaction.
  \item \textbf{Warm SLOAD}: $\approx 100$ gas on subsequent access.
\end{itemize}
Additional fixed costs include the base transaction cost ($21{,}000$
gas), calldata costs ($4$ gas per zero byte, $16$ per non-zero byte),
and event emission overhead.

The critical structural property of this schedule, for the purposes
of this paper, is that the cost of a storage operation depends only
on whether the target slot is cold or warm \emph{within the current
transaction}. It does not depend on the depth of the Merkle Patricia
Trie representing global state, on the size of any contract's
storage, or on the total number of accounts in the network. This is
the property that makes a constant-op contract truly $O(1)$ in gas
cost regardless of how large the registry becomes.

We denote the aggregate fixed overhead per transaction by $\alpha$,
the marginal warm-SSTORE cost by $\beta \approx 5{,}000$, and the
marginal cold-SSTORE cost by $\gamma \approx 22{,}100$.

\subsection{Two Storage Tiers in the EVM}
\label{section:tiers}

The EVM exposes two distinct on-chain storage tiers with different
access, mutability, and cost characteristics. Both are persisted
indefinitely by the network, but they serve different roles in
contract design.

\textbf{Storage slots} are modified via \texttt{SSTORE} and are part
of the world state $\sigma$ maintained by every full node. They are
queryable by other contracts via \texttt{SLOAD}, mutable in
principle (a contract can overwrite a previously written slot), and
expensive ($5{,}000$--$22{,}100$ gas per write under the schedule of
Section~\ref{section:gas-schedule}). Their canonical role is to
carry the protocol's current state in a form directly accessible
from on-chain logic.

\textbf{Event logs} are emitted via the \texttt{LOG} opcodes and are
part of transaction receipts. They are \emph{not} readable by smart
contracts; once emitted, no opcode can modify or delete them, so
they are structurally append-only at the EVM level. Each log is
attested through the receipts root of the block header that
contains its transaction, and is therefore verifiable from chain
headers alone, independent of current world state. Event emission
is substantially cheaper than storage writes (approximately $375$
gas per topic plus calldata-priced data) because the EVM does not
have to support contract-side reads against them.

These two tiers permit a design choice the cost analysis in this
paper turns on: \emph{where does the canonical history of an
append-only registry live?} A storage-canonical design places the
history in slots (typically as a Merkle root commitment), with
events serving as auxiliary notifications. A log-canonical design
places the history in the event stream, with storage slots serving
only as a queryable on-chain index. Both designs are durable; both
are auditable; both are deployable. They differ in cost, in
consumer access pattern, and in what kinds of queries each natively
supports.

The remainder of the paper analyzes one canonical instance of each
choice. The incremental Merkle tree
(Section~\ref{section:imt-primitive}) is a storage-canonical
primitive: the Merkle root in storage is the commitment, and the
event log is auxiliary. The parent-hash directed acyclic graph
(Section~\ref{section:phdag-primitive}) is a log-canonical
primitive: the insertion event carries the complete metadata
required for reconstruction
(Proposition~\ref{prop:reconstruct}), and storage slots serve as a
queryable index supporting append-time integrity checks. The cost
asymmetry between the two primitives, developed formally in
Section~\ref{section:complexity} and empirically in
Section~\ref{section:results}, is a direct consequence of this
architectural difference.

Figure~\ref{fig:architecture} renders this distinction visually.

\begin{figure}[t]
\centering
\includegraphics[width=0.95\linewidth]{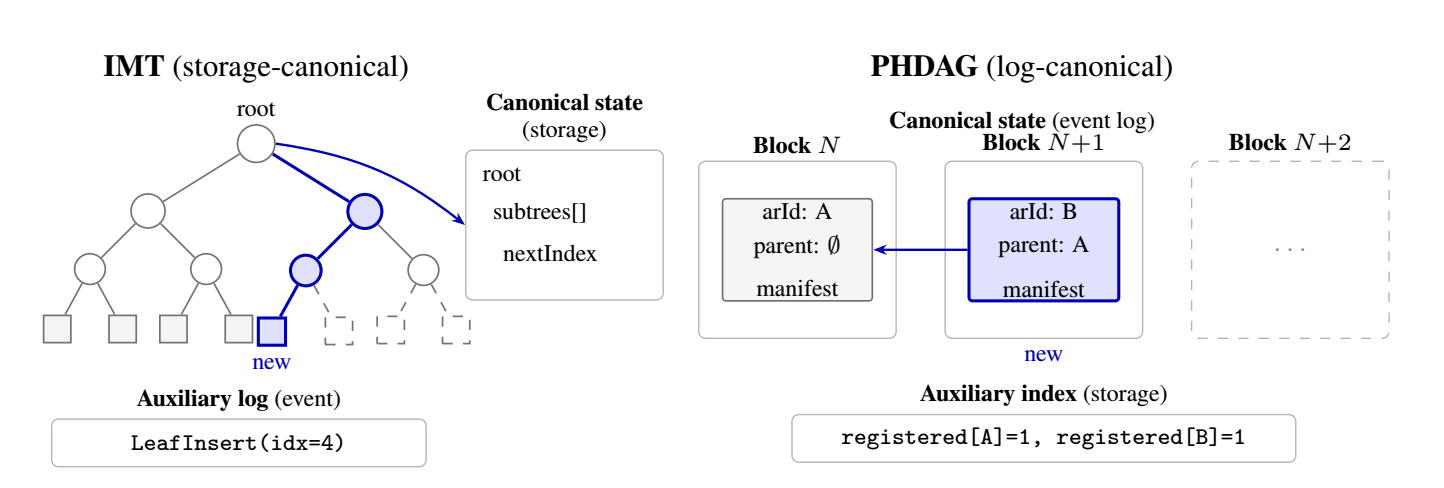}
\caption{Architectural comparison of append-only primitives.
\textbf{Left (IMT, storage-canonical):} an append updates the
frontier path through a fixed-depth binary tree (here $d=3$),
recomputing the root in contract storage. The event log carries an
auxiliary insertion notification. Cost grows with depth.
\textbf{Right (PHDAG, log-canonical):} an append emits a node into
the event log of the block containing its transaction, with parent
references linking nodes across blocks. Contract storage carries
only an index supporting append-time integrity checks. Cost is
constant in depth and in registry size. Accent color marks elements
active during the current append.}
\label{fig:architecture}
\end{figure}

\subsection{The Incremental Merkle Tree Primitive}
\label{section:imt-primitive}

The IMT is a storage-canonical primitive in the sense of
Section~\ref{section:tiers}: the protocol's canonical history is
summarized in a Merkle root persisted to a storage slot, with the
insertion event emitted as an auxiliary notification.
Figure~\ref{fig:architecture} (left panel) shows the structure
visually.

Conceptually, an IMT of depth $d$ is a complete binary tree of
capacity $N = 2^d$ leaves. Internal nodes are hashes of their
children; unpopulated leaves hold a precomputed \emph{zero hash}.
The hash at the root of the tree is a cryptographic commitment to
the entire current state: any change to any leaf changes the root.
This single-value commitment is what makes the IMT useful. A
verifier in possession of the root and a sequence of $d$ sibling
hashes (the \emph{inclusion proof}~\cite{Merkle87}) can confirm
membership of any leaf without reference to the rest of the tree.

An append to an IMT works as follows. A new leaf is placed at the
next unfilled position in leaf-index order, left-to-right. The
contract must then recompute the chain of internal hashes from the
new leaf up to the root. To avoid re-reading the entire tree on
every append, the contract maintains a \emph{frontier}: an array
of length $d$ caching the rightmost filled internal node at each
level. The append walks up the frontier, updating cached nodes as
it goes, and writes the new root to its storage slot.
Figure~\ref{fig:architecture} highlights the frontier path
traversed by a single append in accent color.

The pattern descends from the Merkle tree construction of
Merkle~\cite{Merkle87} and the append-only Merkle log pattern of
Certificate Transparency~\cite{Laurie13}. The fixed-capacity,
frontier-based variant we analyze here was crystallized as a
production Ethereum primitive by the \texttt{MerkleTreeWithHistory}
contract of Tornado Cash~\cite{TornadoCash19}, from which
Semaphore~\cite{Semaphore20} and the
major zero-knowledge rollups~\cite{ZkSync22,Scroll23,Linea23}
inherit the design. We adopt an equivalent minimal reference
implementation in the present work to isolate the primitive's cost
from application-specific overhead.

The cost-relevant property of the IMT is its
\emph{frontier-update pattern}. At each level $k \in \{0, 1, \ldots, d-1\}$ during an
insert at leaf index $i$, the contract writes to
\texttt{filledSubtrees[k]} iff bit $k$ of $i$ is zero; otherwise
it reads \texttt{filledSubtrees[k]} and does not write. In
addition to this per-level conditional write, the contract always
writes to the \texttt{root} and \texttt{nextIndex} slots. This
asymmetric write pattern is the source of IMT's depth-dependent
cost and cost variability.

\subsection{The Parent-Hash DAG Primitive}
\label{section:phdag-primitive}

The PHDAG is a log-canonical primitive in the sense of
Section~\ref{section:tiers}: the protocol's canonical history is
carried in the event stream, with storage slots serving as a
queryable on-chain index that supports append-time integrity
checks. Figure~\ref{fig:architecture} (right panel) shows the
structure visually.

Conceptually, a PHDAG is a directed acyclic graph built
incrementally: each appended node carries a globally unique
identifier and a pointer to a previously-appended parent node,
and the structure grows one node at a time. There is no fixed
geometric shape: nodes may branch arbitrarily, depth is unbounded,
and the parent reference is the only inter-node link. The
structure's integrity comes from the chain of parent references
and from the content commitment (manifest hash) carried by each
node, not from any global root commitment. There is no Merkle
root, and accordingly no succinct membership proof of the kind IMT
provides.

An append to a PHDAG works as follows. The inserting party submits
three values: a globally unique identifier $\mathsf{arId}$, a
parent reference (or the zero value if this is a root node), and a
manifest commitment. The contract writes these into three
mapping-backed storage slots, each keyed by $\mathsf{arId}$, and
emits an insertion event carrying the full metadata. Because
$\mathsf{arId}$ is globally unique by construction, every append
writes to three storage slots that have never been written to
before in the contract's lifetime: three cold SSTOREs. The
storage slots are not where the registry's canonical state lives;
their purpose is to enable the contract to validate, at the next
append, that a referenced parent exists. The canonical state
lives in the event log, distributed across the blocks containing
the registry's inserts. Figure~\ref{fig:architecture} shows two
nodes appended in adjacent blocks, with the parent reference
linking them across the block boundary.

The pattern descends from the parent-hash content-addressed
structures of git~\cite{Git} and IPFS~\cite{IPFS}, both of which
build DAGs in which each node references its parent(s) by content
hash and the structure's integrity is enforced by hash-chain
validation rather than by a global commitment. The PHDAG primitive
analyzed here adapts this pattern to the on-chain setting and was
introduced as the data-structural core of the AnchorRegistry
provenance registry of Moore~\cite{Moore26Trustless}, where each
node carries application-specific metadata not considered in the
present analysis. For our cost analysis we treat the primitive in
its minimal form: three SSTOREs per append to unique keys, plus
one SSTORE to the shared counter slot, plus an insertion event.
The reference contract used for our empirical work implements
exactly this minimal form.

The cost-relevant property of the PHDAG is that each append's
storage writes target previously-untouched slots, with no
dependence on any geometric parameter (depth, capacity, registry
size). This is the property that yields the $O(1)$ append
complexity formalized in Section~\ref{section:complexity}.

\subsection{Notation}

Throughout, we denote:
\begin{itemize}
  \item $n$: the number of anchors currently registered in the registry
  \item $d$: tree depth parameter for IMT, fixed at contract
        deployment. For PHDAG the notion of depth applies
        graph-theoretically to nodes (distance from root along
        parent references) but is not a parameter of the contract
        or of any cost expression in this paper.
  \item $N = 2^d$: IMT capacity
  \item $i$: leaf index of a given insert, uniformly distributed in $[0, N)$
  \item $g$: total gas consumed by a single insert, treated as a random variable
  \item $\alpha$: fixed per-transaction overhead (base tx + calldata + event)
  \item $\beta$: marginal cost per warm SSTORE ($\approx 5{,}000$ gas)
  \item $\gamma$: marginal cost per cold SSTORE ($\approx 22{,}100$ gas)
  \item $c_0$: fixed per-insert overhead for the IMT cost model
        (base tx, calldata, event, and the unconditional
        \texttt{root}/\texttt{nextIndex} writes)
  \item $c_L$: marginal cost of an IMT \emph{write level}
        (left-child level: SSTORE + SLOAD + keccak256 + overhead)
  \item $c_R$: marginal cost of an IMT \emph{read level}
        (right-child level: SLOAD + keccak256 + overhead)
  \item $W(i,d)$: number of write levels for an insert at leaf index
        $i$, equal to $d - \mathrm{pop}(i)$ (Lemma~\ref{lemma:writes})
  \item $\mathrm{pop}(x)$: the Hamming weight (popcount) of $x$
  \item $|V|$: the number of nodes in a provenance tree $V$
\end{itemize}

\subsection{The Ledger as Coordination Medium}
\label{section:coordination-medium}

A property orthogonal to the cost analysis but central to the role
of append-only registries in practice is the function of the ledger
as a \emph{coordination medium}. Paredes
Garc\'ia~\cite{Paredes26} develops this perspective formally as
\emph{ledger-state stigmergy}: autonomous agents (bots, indexers,
downstream contracts, AI tools) coordinate not by exchanging
messages but by reading shared ledger state and reacting to traces
left by prior actions. Smart contracts in this reading are
coordination artifacts embedded in the medium; each append leaves a
persistent, cryptographically authenticated, globally visible trace
that subsequent agents may act upon. The framework identifies three
recurring patterns for such coordination: \emph{State-Flag} (agents
poll storage variables), \emph{Event-Signal} (agents subscribe to
event logs), and \emph{Threshold-Trigger} (agents act when
quantitative state variables cross a boundary). We adopt this
vocabulary throughout: the reconstruction property of PHDAG
(Section~\ref{section:complexity}) is presented as a structural
guarantee, and its coordination-medium consequences are developed
in Section~\ref{section:discussion}.

\section{Formal Complexity Analysis}
\label{section:complexity}

We begin by establishing the structural complexity of each primitive,
independent of any stochastic or empirical refinement. The results of
this section are statements about the \emph{number} of EVM operations
executed per append, from which gas-cost asymptotics follow under the
fixed-per-opcode pricing of EIP-2929/3529.

\subsection{PHDAG Append is $O(1)$}

\begin{lemma}[Bounded On-Chain Operations]
\label{lemma:bounded}
The PHDAG append function executes a fixed, bounded number of EVM
operations, independent of the registry size $n$ and of any notion
of tree depth.
\end{lemma}

\begin{proof}
The append function executes exactly the following steps:
\begin{enumerate}
  \item Input validation (length and non-empty checks on the
        identifier and manifest commitment): $O(1)$.
  \item A bounded number of lookups into the \texttt{registered}
        mapping: one \texttt{SLOAD} to confirm the new identifier is
        unused, and, for a non-root append, one further \texttt{SLOAD}
        to confirm the referenced parent exists (a root append, whose
        parent reference is the zero value, skips the second by
        short-circuit evaluation). At most two \texttt{SLOAD}s, each
        with constant gas cost.
  \item For the general (non-root) append, three \texttt{SSTORE}
        operations writing the identifier, parent reference, and
        manifest commitment to slots keyed by the unique
        $\mathsf{arId}$. When the parent reference is the zero value
        (a root append, the case measured in
        Section~\ref{section:results}), the write to
        \texttt{parents[arId]} stores zero to an already-zero slot and
        is billed at the no-op rate rather than as a zero-to-nonzero
        set, leaving two cold sets ($2\gamma$) plus one zero-store;
        a non-root append writes a nonzero parent and incurs the third
        set ($3\gamma$). Either way the count is a fixed constant
        independent of $n$ and of depth.
  \item One \texttt{SSTORE} updating the monotonic counter, preceded
        by one \texttt{SLOAD} of the counter slot: a single read and
        a single warm write at $\beta$ gas.
  \item Emission of one insertion event with a fixed number of
        topics and a bounded data payload.
\end{enumerate}
Each of the five steps performs a fixed number of EVM operations:
step 1, a constant number of validation opcodes; step 2, at most two
SLOADs; step 3, three SSTOREs (two cold sets and one zero-store for a
root append, three cold sets for a non-root append); step 4, one SLOAD and one
SSTORE; step 5, exactly one event emission. Thus no step contains
a loop, iteration, or traversal whose count varies with $n$ or with
any property of the registry. The total opcode count is therefore
the sum of these per-step constants, which is itself a constant
$c_1$ independent of $n$ and of any depth. \qed
\end{proof}

The insertion event in step 5 above is not auxiliary to the storage
writes in steps 3 and 4; under the log-canonical design of
Section~\ref{section:tiers}, it carries the canonical metadata that
the storage slots merely index. Both are necessary, and both are
constant-cost.

\begin{lemma}[Registry-Size Invariance]
\label{lemma:nodep}
The gas cost of a PHDAG append is invariant to the global registry
size $n$. Formally, treating the per-append storage gas $G(n)$ as
a function of the registry's pre-append cardinality $n \in \mathbb{N}_0$,
\[
  \frac{\partial G}{\partial n} \;=\; 0.
\]
\end{lemma}

\begin{proof}
The proof proceeds in four parts: a formal model of the contract's
storage footprint, a formal model of the EVM's per-transaction
access list under EIP-2929/3529, a step-by-step execution trace of
the append, and the resulting invariance.

\textbf{(i) Contract storage footprint.}
The persistent state of the deployed PHDAG contract is modeled as a
key-value mapping $\sigma : \mathcal{K} \to \mathcal{V}$, where
$\mathcal{K} = \mathcal{V} = \{0,1\}^{256}$. For a registry of size
$n$, the occupied storage footprint $\mathcal{O}(n) \subset \mathcal{K}$
partitions into four disjoint subsets:
\[
  \mathcal{O}(n) \;=\; S_{\text{reg}}(n) \,\cup\, S_{\text{par}}(n) \,\cup\, S_{\text{man}}(n) \,\cup\, \{k_{\text{ctr}}\},
\]
where $S_{\text{reg}}, S_{\text{par}}, S_{\text{man}}$ are the
slot sets for the \texttt{registered}, \texttt{parents}, and
\texttt{manifests} mappings respectively, each of cardinality $n$,
and $k_{\text{ctr}}$ is the fixed slot holding the monotonic
counter \texttt{count}. Thus $|\mathcal{O}(n)| = 3n + 1$, growing
linearly with $n$.

For a mapping declared at storage position $p$, the address of the
entry under key $\kappa$ is given by the Solidity storage layout
specification:
\[
  k(\kappa, p) \;=\; \mathrm{keccak256}(\kappa \,\|\, p).
\]

\textbf{(ii) Per-transaction access list.}
Under the EIP-2929/3529 gas schedule (Section~\ref{section:gas-schedule}),
each transaction maintains a transient \emph{access list}
$\mathcal{A}_t \subset \mathcal{K}$ recording the slots touched at
or before execution step $t$. The cost of a write to slot $k$ at
step $t$ is:
\[
  \mathrm{Cost}_{\mathrm{SSTORE}}(k, \mathcal{A}_t) \;=\;
  \begin{cases}
    \gamma & k \notin \mathcal{A}_t \quad \text{(cold)} \\
    \beta  & k \in \mathcal{A}_t \quad \text{(warm)}
  \end{cases}
\]
with $\gamma \approx 22{,}100$ and $\beta \approx 5{,}000$ gas. The
cost of a read (SLOAD) at slot $k$ has the analogous form with
constants $\gamma_{\mathrm{R}} \approx 2{,}100$ and
$\beta_{\mathrm{R}} \approx 100$. After any access at slot $k$, the
list updates monotonically: $\mathcal{A}_{t+1} = \mathcal{A}_t \cup \{k\}$.
Critically, $\mathcal{A}_0 = \emptyset$ at the start of every
transaction: the access list does not persist across transactions
and does not depend on $\mathcal{O}(n)$.

\textbf{(iii) Execution trace.}
Let the append be invoked with a fresh identifier $\kappa^*$
and parent reference $\kappa_{\mathrm{par}}$. By construction
(Section~\ref{section:phdag-primitive}), $\kappa^*$ has not been
used in any prior append, so
$k(\kappa^*, p_{\mathrm{reg}}), k(\kappa^*, p_{\mathrm{par}}), k(\kappa^*, p_{\mathrm{man}}) \notin \mathcal{O}(n)$
for the respective mapping positions $p_{\mathrm{reg}}, p_{\mathrm{par}}, p_{\mathrm{man}}$.
The contract executes the following storage-touching steps:
\begin{enumerate}
  \item[$t=0$] \textbf{Duplicate check.} SLOAD on
        $k(\kappa^*, p_{\mathrm{reg}})$ to confirm the new identifier
        is unused. This slot is not in $\mathcal{A}_0 = \emptyset$,
        so the access is cold, costing $\gamma_{\mathrm{R}}$.
        $\mathcal{A}_1 = \{k(\kappa^*, p_{\mathrm{reg}})\}$.
  \item[$t=1$] \textbf{Parent existence check.} SLOAD on
        $k(\kappa_{\mathrm{par}}, p_{\mathrm{reg}})$. For a non-root
        append this slot is not in $\mathcal{A}_1$, so the access is
        cold, costing $\gamma_{\mathrm{R}}$; for a root append
        ($\kappa_{\mathrm{par}} = 0$) the check is skipped by
        short-circuit evaluation and this SLOAD does not occur. We
        trace the non-root case, which is the general one.
        $\mathcal{A}_2 = \mathcal{A}_1 \cup \{k(\kappa_{\mathrm{par}}, p_{\mathrm{reg}})\}$.
  \item[$t=2$] \textbf{Write \texttt{registered[arId]}.} SSTORE on
        $k(\kappa^*, p_{\mathrm{reg}})$. This slot was warmed by the
        duplicate check at $t=0$, but holds the zero value until now;
        the write is a zero-to-nonzero allocation billed at the cold
        SSTORE rate $\gamma$.
        $\mathcal{A}_3 = \mathcal{A}_2$ (slot already present).
  \item[$t=3$] \textbf{Write \texttt{parents[arId]}.} SSTORE on
        $k(\kappa^*, p_{\mathrm{par}})$. Different mapping position,
        so the keccak256 output is distinct; not in
        $\mathcal{A}_3$. For a root append ($\kappa_{\mathrm{par}} = 0$,
        the case measured in Section~\ref{section:results}) this stores
        zero to an already-zero slot: a zero-to-zero no-op billed at
        the warm/no-op rate, not a cold allocation. For a non-root
        append the value written is nonzero, giving a cold allocation
        $\gamma$. We carry the non-root cost $\gamma$ as the general
        case; the root case substitutes the smaller zero-store cost.
        $\mathcal{A}_4 = \mathcal{A}_3 \cup \{k(\kappa^*, p_{\mathrm{par}})\}$.
  \item[$t=4$] \textbf{Write \texttt{manifests[arId]}.} SSTORE on
        $k(\kappa^*, p_{\mathrm{man}})$. Distinct from prior slots
        by the same argument; not in $\mathcal{A}_4$. Cold
        allocation: $\gamma$.
        $\mathcal{A}_5 = \mathcal{A}_4 \cup \{k(\kappa^*, p_{\mathrm{man}})\}$.
  \item[$t=5$] \textbf{Read counter.} SLOAD on $k_{\mathrm{ctr}}$
        (compiled from \texttt{count += 1}). Not in $\mathcal{A}_5$,
        so cold: $\gamma_{\mathrm{R}}$.
        $\mathcal{A}_6 = \mathcal{A}_5 \cup \{k_{\mathrm{ctr}}\}$.
  \item[$t=6$] \textbf{Write counter.} SSTORE on $k_{\mathrm{ctr}}$.
        Now in $\mathcal{A}_6$ (warmed by the immediately preceding
        SLOAD). Warm: $\beta$.
\end{enumerate}

\textbf{(iv) Invariance.}
The total storage gas of the append is
\[
  G(n) \;=\; \underbrace{3\gamma_{\mathrm{R}}}_{\text{3 cold SLOADs}} \;+\; \underbrace{3\gamma}_{\text{3 cold SSTOREs}} \;+\; \underbrace{\beta}_{\text{1 warm SSTORE}}
\]
for a non-root append, with the root case omitting one SLOAD
($2\gamma_{\mathrm{R}}$ in place of $3\gamma_{\mathrm{R}}$) and
replacing one cold set with a zero-store (two cold sets in place of
three, since the zero parent reference is stored to an already-zero
slot).
Each term is a protocol constant; none is a function of $n$. The
access list $\mathcal{A}_t$ at each step is a function of the
transaction's own prior steps, not of $\mathcal{O}(n)$, because
$\mathcal{A}_0 = \emptyset$ and the only insertions into $\mathcal{A}$
during the append are the slots the append itself touches. The
cold/warm classification of each step is therefore determined by
the transaction's internal structure and is independent of $n$.
Consequently,
\[
  \frac{\partial G}{\partial n} \;=\; 0, \qquad \forall n \in \mathbb{N}_0.
\]
The remaining per-transaction costs (base transaction, calldata,
opcode execution outside storage, event emission) are functions of
the transaction's input length and opcode sequence, both fixed for
a given append, and are likewise independent of $n$. The total
cost $g_{\mathrm{PHDAG}}$ is therefore a constant function of $n$.
\qed
\end{proof}

\begin{theorem}[$O(1)$ Gas Complexity of PHDAG Append]
\label{thm:phdag}
The gas cost $g_{\mathrm{PHDAG}}$ of a PHDAG append satisfies
\[
g_{\mathrm{PHDAG}} \;\leq\; C \qquad \text{for some constant } C > 0,
\]
and in particular $g_{\mathrm{PHDAG}} = O(1)$ with respect to both
the global registry size $n$ and any notion of tree depth.
\end{theorem}

\begin{proof}
By Lemma~\ref{lemma:bounded}, the number of EVM operations per append
is bounded by a constant $c_1$. By Lemma~\ref{lemma:nodep}, the gas
cost of each such operation is independent of $n$. Therefore the
total gas cost is bounded by a constant $C$ proportional to $c_1$.
PHDAG has no tree-depth parameter, so depth-invariance holds vacuously.
\qed
\end{proof}

Writing out the constant explicitly:
\begin{equation}
g_{\mathrm{PHDAG}} \;\approx\; \alpha_{\mathrm{PHDAG}} + 3\gamma + 3\gamma_{\mathrm{R}} + \beta,
\label{eq:phdag}
\end{equation}
for a non-root append, where $\alpha_{\mathrm{PHDAG}}$ bundles the
base transaction cost, calldata, event emission, and input
validation; $3\gamma$ is the three cold SSTOREs writing to unique
keys; $3\gamma_{\mathrm{R}}$ is the three cold SLOADs of
Lemma~\ref{lemma:nodep} (duplicate check, parent existence check,
counter read); and $\beta$ is the warm SSTORE to the counter. The
root append benchmarked in Section~\ref{section:results} differs in
two fixed terms (Lemma~\ref{lemma:nodep}): it omits the
parent-existence SLOAD ($2\gamma_{\mathrm{R}}$ in place of
$3\gamma_{\mathrm{R}}$), and its write to \texttt{parents[arId]}
stores the zero parent reference to an already-zero slot, a
zero-to-zero no-op in place of a third cold set ($2\gamma$ in place
of $3\gamma$). The expression above is therefore the count of
structural operations, not a closed-form gas total: the per-operation
prices $\gamma$, $\gamma_{\mathrm{R}}$, $\beta$ are upper nominal
figures, and the value-transition cases (a slot pre-warmed by an
earlier SLOAD, a zero-to-nonzero set versus a zero-store) are priced
distinctly under EIP-2200/2929/3529. We do not propagate these into
a hand-derived constant; the per-append constant is read directly
from the measurement in Section~\ref{section:results}. What the
expression establishes is the property that matters for
Theorem~\ref{thm:phdag}: a fixed set of storage operations,
independent of $n$ and of depth, hence a fixed gas cost.

\subsection{IMT Append is $O(d)$}

By the structural write pattern described in
Section~\ref{section:background}, IMT append executes a loop of
length $d$, with each iteration performing either an \texttt{SLOAD}
or an \texttt{SSTORE} depending on the leaf index parity. The number
of EVM operations per append therefore scales linearly in $d$.

\begin{proposition}[$O(d)$ Gas Complexity of IMT Append]
\label{prop:imt}
The gas cost $g_{\mathrm{IMT}}$ of an IMT append at depth $d$
satisfies
\[
g_{\mathrm{IMT}} \;=\; O(d)
\]
with respect to tree depth.
\end{proposition}

The precise distribution of per-append cost within the $O(d)$
envelope is the subject of the stochastic analysis in
Section~\ref{section:stochastic}.

\subsection{Off-Chain Reconstruction is $O(|V|)$}

A property orthogonal to per-append cost is the complexity of
reconstructing the complete registry from public blockchain data
alone. For PHDAG, the registry state is fully encoded in the
sequence of insertion events; no off-chain index is required for
correctness.

\begin{proposition}[$O(|V|)$ Trustless Reconstruction]
\label{prop:reconstruct}
Let $V$ be the set of nodes registered in a deployed PHDAG contract
instance up to the current block height. Given the deployment block
$B_{\text{start}}$ and the contract address, the complete directed
acyclic graph $G = (V, E)$ is reconstructible from public event
logs alone in time $\mathcal{O}(|V|)$ and space $\mathcal{O}(|V|)$,
with no dependency on any off-chain index or trusted service. The
edges in $E$ are oriented parent$\to$child, consistent with the
ancestry-integrity convention of~\cite{Moore26Trustless}.
\end{proposition}

\begin{proof}
The proof proceeds in four parts: the event log model, an explicit
reconstruction algorithm, time complexity, and space complexity.

\textbf{(i) Event log model.}
Each invocation of the PHDAG append function emits exactly one
\texttt{Appended} event carrying the full metadata of the registered
node. Model the $j$-th such event as the tuple
\[
  L_j = \langle \kappa_j, \mu_j, \kappa_{\mathrm{par},j} \rangle,
\]
where $\kappa_j$ is the unique identifier, $\mu_j$ is the manifest
commitment, and $\kappa_{\mathrm{par},j}$ is either the parent's
identifier or the sentinel $\emptyset$ if $\kappa_j$ is a root node.
A single \texttt{eth\_getLogs} call filtered by the event signature
hash $\mathcal{T}_0 = \mathrm{keccak256}(\text{``Appended(bytes32,bytes32,bytes32)''})$
over the block range $[B_{\text{start}}, B_{\text{current}}]$
returns the chronologically ordered sequence
\[
  \mathcal{S} = [L_1, L_2, \ldots, L_m].
\]
The contract reverts on duplicate identifiers (Section~\ref{section:phdag-primitive}),
so each event corresponds to a distinct node, and the append-only
nature of the ledger guarantees no event is lost. Therefore
$m = |V|$.

\textbf{(ii) Reconstruction algorithm.}
Initialize an empty graph $G = (V, E)$ with $V = \emptyset$,
$E = \emptyset$, and an empty hash map $\mathcal{M}$ (with
amortized $\mathcal{O}(1)$ insertion). Process $\mathcal{S}$ in a
single pass:
\begin{algorithmic}[1]
\For{\textbf{each} $L_j = \langle \kappa_j, \mu_j, \kappa_{\mathrm{par},j} \rangle \in \mathcal{S}$}
    \State $V \gets V \cup \{\kappa_j\}$
    \State $\mathcal{M}[\kappa_j] \gets \mu_j$
    \If{$\kappa_{\mathrm{par},j} \neq \emptyset$}
        \State $E \gets E \cup \{(\kappa_{\mathrm{par},j}, \kappa_j)\}$
    \EndIf
\EndFor
\end{algorithmic}
The edge ordering $(\kappa_{\mathrm{par},j}, \kappa_j)$ encodes the
provenance direction parent$\to$child.

\textbf{(iii) Time complexity.}
Fetching $\mathcal{S}$ is an $\mathcal{O}(|V|)$ data transfer over
the JSON-RPC interface. Each iteration of the loop performs a
constant number of operations: one vertex insertion ($c_V$), one
hash-map insertion ($c_M$, amortized constant), one conditional
test, and at most one edge insertion ($c_E$). Letting
$C_{\mathrm{loop}} = c_V + c_M + c_E$,
\[
  T(|V|) \;=\; \sum_{j=1}^{|V|} C_{\mathrm{loop}} \;=\; C_{\mathrm{loop}} \cdot |V| \;=\; \mathcal{O}(|V|).
\]

\textbf{(iv) Space complexity.}
The algorithm stores each identifier once in $V$ and once as a key
in $\mathcal{M}$, and each non-root node contributes one edge to
$E$. Because every node has at most one parent reference,
$|E| \leq |V|$. The total space is therefore
\[
  S(|V|) \;=\; |V| + |V| + |E| \;\leq\; 3|V| \;=\; \mathcal{O}(|V|).
\]
The graph is fully recovered in linear time and space without
invoking any off-chain index or trusted service. \qed
\end{proof}

This property is significant for the robustness of PHDAG-based
registries: the integrity of the registry is not coupled to any
off-chain service. A complete re-sync can be performed by any
party with read access to the blockchain, with time complexity
linear in the registry size and no other dependencies.

\section{A Stochastic Model of IMT Per-Insert Cost}
\label{section:stochastic}

Theorem~\ref{thm:phdag} and Proposition~\ref{prop:imt} establish the
structural complexity classes of the two primitives. We now refine
the IMT analysis by modeling the per-insert cost as a random variable
whose distribution is determined by the binary representation of the
leaf index.

\subsection{The Frontier Write Distribution}

\begin{lemma}[Frontier Write Count]
\label{lemma:writes}
For an insert at leaf index $i \in [0, 2^d)$, the number of
\texttt{filledSubtrees[k]} slots written during the frontier update
loop is
\[
  W(i, d) \;=\; d - \mathrm{pop}(i \bmod 2^d),
\]
where $\mathrm{pop}(\cdot)$ denotes the Hamming weight.
\end{lemma}

\begin{proof}
The proof proceeds in four parts: the frontier state model, the
recurrence governing the loop's index variable, the mutation
condition at each level, and the resulting Hamming-weight identity.

\textbf{(i) Frontier state.}
An IMT of depth $d$ maintains an array
$\mathbf{f} = [f_0, f_1, \dots, f_{d-1}]$ holding the rightmost
filled internal node at each level. The insertion routine traverses
levels $k = 0, 1, \dots, d-1$ from leaf to root, conditionally
updating $\mathbf{f}$ at each level. Let $b_k(i)$ denote the $k$-th
bit of $i$ in its unique base-2 expansion:
\[
  i = \sum_{k=0}^{d-1} b_k(i) \cdot 2^k, \qquad b_k(i) = \lfloor i / 2^k \rfloor \bmod 2.
\]

\textbf{(ii) Loop recurrence.}
The contract maintains a local variable $x_k$ tracking the active
node index at level $k$, with initial value and update rule:
\[
  x_0 = i, \qquad x_{k+1} = \lfloor x_k / 2 \rfloor.
\]
Unrolling the recurrence to depth $k$ yields the closed form
\[
  x_k = \lfloor i / 2^k \rfloor.
\]

\textbf{(iii) Mutation condition.}
At level $k$ the contract writes to $f_k$ iff $x_k$ corresponds to
a left child of its parent in the binary tree, equivalently, iff
$x_k$ is even. When $x_k$ is odd, the loop instead reads $f_k$ to
combine with $x_k$'s sibling and ascends without writing. The
write condition is therefore
\[
  \text{Write}_k \iff x_k \equiv 0 \pmod 2 \iff \lfloor i / 2^k \rfloor \bmod 2 = 0 \iff b_k(i) = 0.
\]
Let $\mathbb{I}_k = \mathbb{I}(b_k(i) = 0) = 1 - b_k(i)$ be the
indicator of a write at level $k$.

\textbf{(iv) Total write count.}
Summing the indicator over all $d$ levels:
\[
  W(i, d) \;=\; \sum_{k=0}^{d-1} \mathbb{I}_k \;=\; \sum_{k=0}^{d-1} \bigl(1 - b_k(i)\bigr) \;=\; d - \sum_{k=0}^{d-1} b_k(i).
\]
The remaining sum is the count of $1$-bits in the low $d$ bits of
$i$, which is by definition $\mathrm{pop}(i \bmod 2^d)$. Therefore
\[
  W(i, d) \;=\; d - \mathrm{pop}(i \bmod 2^d). \qed
\]
\end{proof}

\begin{corollary}[Total Warm-SSTORE Count]
\label{cor:sstore}
The total number of warm SSTORE operations performed per IMT insert
is
\[
S(i, d) \;=\; W(i, d) + 2 \;=\; d - \mathrm{pop}(i \bmod 2^d) + 2,
\]
where the additive $2$ accounts for the unconditional writes to
\texttt{root} and \texttt{nextIndex}.
\end{corollary}

\subsection{Distribution of Cost}

Treating the leaf index $i$ as uniformly distributed on $[0, 2^d)$,
the popcount $P = \mathrm{pop}(i)$ is distributed as a sum of $d$
independent Bernoulli($1/2$) random variables:
\[
P \sim \mathrm{Binomial}(d, 1/2).
\]

\begin{remark}[On the uniform-index assumption]
\label{rem:uniform}
Real-world appends populate the IMT sequentially in leaf-index order
$i = 0, 1, 2, \ldots$, which is not uniform. We model $i$ as
uniformly distributed in $[0, 2^d)$ because the cost of an append
depends on $i$ only through its low-order bit pattern (Lemma~\ref{lemma:writes}),
and across any contiguous span of $2^d$ sequential indices the
popcount distribution coincides exactly with the binomial. The
uniform model is therefore the correct \emph{lifetime-average}
distribution for the per-insert cost: it is the mean over a full
epoch of $2^d$ sequential appends, not an idealization that
abstracts away the order of insertion.

This lifetime-average assumption should be distinguished from the
index distribution \emph{sampled by our experiments}. The depth
sweep (Section~\ref{section:methodology}) measures only the first
$\max(d,2)$ appends at each depth, i.e.\ the small indices
$i = 1, 2, 3, \ldots$, which have low Hamming weight and hence an
above-average write-level count $W = d - \mathrm{pop}(i)$. The sweep
is thus a write-heavy sub-sample of the lifetime distribution, not a
uniform draw from it. The two coincide only in the limit of a fully
populated tree. This distinction is load-bearing for interpreting the
measured mean slope and crossover (Sections~\ref{section:crossover}
and~\ref{section:agreement}): the model evaluated at the uniform
distribution and the model evaluated at the sweep's sampled indices
are the same model, and both are reported. Non-uniform write
schedules in the wild (skipping indices, adversarial ordering) are
excluded by the monotonic counter in the canonical IMT
implementation; an inserting party cannot choose its leaf index.
\end{remark}

\begin{theorem}[Moments of Warm-SSTORE Count]
\label{thm:moments}
Under uniform sampling of the leaf index $i$ over the tree capacity
index space $\Omega = \{0, 1, \ldots, 2^d - 1\}$, the per-insert
warm-SSTORE count $S$ defined in Corollary~\ref{cor:sstore} has
moments
\begin{align}
\mathbb{E}[S]   &= \tfrac{d}{2} + 2, \\
\mathrm{Var}[S] &= \tfrac{d}{4}.
\end{align}
\end{theorem}

\begin{proof}
The proof proceeds in four parts: the probability space, the
decomposition of the leaf index into independent Bernoulli bits,
the resulting binomial distribution of the popcount, and the
application of linearity and scale-translation to compute the
moments of $S$.

\textbf{(i) Probability space.}
Let $(\Omega, \mathcal{F}, \mathbb{P})$ be a finite probability
space with $\Omega = \{0, 1, \ldots, 2^d - 1\}$,
$\mathcal{F} = 2^{\Omega}$, and the uniform probability measure
\[
  \mathbb{P}(I = i) = \frac{1}{2^d}, \qquad \forall i \in \Omega.
\]
By Corollary~\ref{cor:sstore}, the warm-SSTORE count is the
deterministic transformation
\[
  S(I, d) = d - \mathrm{pop}(I) + 2.
\]

\textbf{(ii) Bit-level decomposition.}
Define $d$ random variables $B_0, B_1, \ldots, B_{d-1}$ as the
binary digits of $I$:
\[
  I = \sum_{k=0}^{d-1} B_k \cdot 2^k, \qquad B_k = \lfloor I / 2^k \rfloor \bmod 2.
\]
For each $k$, exactly half of the integers in $\Omega$ have bit
$k$ set, so
\[
  \mathbb{P}(B_k = 1) = \frac{2^{d-1}}{2^d} = \frac{1}{2},
\]
giving $B_k \sim \mathrm{Bernoulli}(1/2)$. Moreover, the uniform
measure on $\Omega$ factorizes exactly as a product of $d$
Bernoulli$(1/2)$ measures: for any bit pattern
$(b_0, \ldots, b_{d-1}) \in \{0,1\}^d$,
\[
  \mathbb{P}(B_0 = b_0, \ldots, B_{d-1} = b_{d-1}) = \frac{1}{2^d} = \prod_{k=0}^{d-1} \mathbb{P}(B_k = b_k).
\]
Therefore the bits $\{B_k\}_{k=0}^{d-1}$ are i.i.d.\ Bernoulli$(1/2)$.

\textbf{(iii) Binomial distribution of the popcount.}
The Hamming weight is the sum of the indicator bits:
\[
  P = \mathrm{pop}(I) = \sum_{k=0}^{d-1} B_k.
\]
As a sum of $d$ i.i.d.\ Bernoulli$(1/2)$ variables, $P$ follows the
binomial distribution
\[
  P \sim \mathrm{Binomial}(d, 1/2),
\]
with first and second moments
\[
  \mathbb{E}[P] = d \cdot \tfrac{1}{2} = \tfrac{d}{2}, \qquad
  \mathrm{Var}[P] = d \cdot \tfrac{1}{2} \cdot \tfrac{1}{2} = \tfrac{d}{4}.
\]

\textbf{(iv) Moments of $S$.}
Applying linearity of expectation to $S = d - P + 2$:
\[
  \mathbb{E}[S] = d - \mathbb{E}[P] + 2 = d - \tfrac{d}{2} + 2 = \tfrac{d}{2} + 2.
\]
Applying the scale-translation rule
$\mathrm{Var}[aX + b] = a^2 \mathrm{Var}[X]$ with $a = -1$ and
$b = d + 2$:
\[
  \mathrm{Var}[S] = (-1)^2 \mathrm{Var}[P] = \mathrm{Var}[P] = \tfrac{d}{4}.
\]
This establishes both moments. \qed
\end{proof}

These are the moments of the write count itself. To turn them into
gas, we attach a cost to each level type and propagate the
distribution of $W$ through the resulting affine cost expression;
this is the content of the next theorem.

\begin{theorem}[Moments of Per-Insert Gas Cost]
\label{thm:cost}
The per-insert gas cost $g$ of an IMT append decomposes by level
type. Each level of the frontier loop is either a \emph{write level}
(a left-child level, where $\mathtt{filledSubtrees}[k]$ is updated)
or a \emph{read level} (a right-child level, where the slot is read
but not written). Writing $W = W(i,d)$ for the write-level count of
Lemma~\ref{lemma:writes}, the cost is
\[
  g \;=\; c_0 + c_L \cdot W + c_R \cdot (d - W),
\]
where $c_0$ is the fixed per-transaction overhead (base transaction,
calldata, event emission, and the unconditional \texttt{root} and
\texttt{nextIndex} writes), $c_L$ is the marginal cost of a write
level, and $c_R$ is the marginal cost of a read level. Under uniform
leaf-index sampling, with $\mathbb{E}[W] = d/2$ and
$\mathrm{Var}[W] = d/4$ (Theorem~\ref{thm:moments}, via
$W = d - P$),
\begin{align}
\mathbb{E}[g]   &= c_0 + \tfrac{d}{2}\,(c_L + c_R), \label{eq:meang} \\
\mathrm{Var}[g] &= (c_L - c_R)^2 \cdot \tfrac{d}{4}, \label{eq:varg}
\end{align}
and the standard deviation is $\sqrt{\mathrm{Var}[g]} = |c_L - c_R|\,\sqrt{d}/2$.
\end{theorem}

\begin{proof}
The level-type decomposition is exact: the frontier loop runs $d$
iterations, of which exactly $W$ are write levels and $d - W$ are
read levels (Lemma~\ref{lemma:writes}). Each write level executes
the same opcode sequence (one SSTORE updating a slot that already
holds a nonzero value, one SLOAD of the precomputed zero hash, one
\texttt{keccak256}, and fixed loop overhead), so it has a fixed
marginal cost $c_L$; each read level executes one SLOAD, one
\texttt{keccak256}, and the same loop overhead, with fixed marginal
cost $c_R$. The per-transaction remainder is collected into $c_0$.
This gives $g = c_0 + c_L W + c_R(d - W) = c_0 + c_R d + (c_L - c_R)W$,
an affine function of the single random variable $W$. Applying
linearity of expectation and the scale-translation rule
$\mathrm{Var}[aX + b] = a^2 \mathrm{Var}[X]$ with the moments of
$W$ yields Equations~\ref{eq:meang} and~\ref{eq:varg}. \qed
\end{proof}

The constants $c_0$, $c_L$, and $c_R$ are determined by the EVM gas
schedule and the contract's opcode sequence.\footnote{A write level
combines an SSTORE to an already-nonzero slot (the frontier slot
holds a nonzero hash from contract initialization onward, so the
update is billed as a value-change reset rather than a
zero-to-nonzero allocation), a cold SLOAD of the zero-hash table,
a \texttt{keccak256} over two words, and loop overhead. A read level
omits the SSTORE. The cold/warm and zero/nonzero pricing of
EIP-2929/3529 makes both per-level costs constants independent of
depth and registry size.} Rather than fix them from first principles,
we calibrate them against the empirical measurements
(Section~\ref{section:results}), where a regression of per-insert
gas on $(W, d - W)$ recovers $c_L \approx 7{,}950$, $c_R \approx 2{,}750$,
and $c_0 \approx 37{,}800$ gas with coefficient of determination
$R^2 > 0.999$, confirming both the level-type decomposition and the
constancy of the per-level costs.

\begin{remark}[Relation to the warm-SSTORE count]
\label{rem:sstore-vs-levels}
The warm-SSTORE count $S = W + 2$ of Corollary~\ref{cor:sstore}
counts only the storage writes. It is not by itself a cost model,
because the read levels also carry depth-linear cost ($c_R$ per
level) through their SLOAD and \texttt{keccak256} operations. A
cost model of the form $g = \alpha + \beta S$, which prices only the
SSTOREs and treats the read levels as part of the fixed overhead,
undercounts the per-depth cost and overstates the crossover depth
(Section~\ref{section:crossover}). The level-type decomposition
above is the corrected model.
\end{remark}

\subsection{Asymptotic Behavior}

The mean cost $\mathbb{E}[g]$ scales linearly in $d$ with slope
$(c_L + c_R)/2$. The standard deviation $\sqrt{\mathrm{Var}[g]}$
scales as $|c_L - c_R|\sqrt{d}/2$, sub-linearly in $d$ but still
growing in absolute terms. Two features of these expressions are
worth drawing out.

First, the variance is governed by the \emph{difference} $c_L - c_R$
between the two per-level costs, not by either cost alone. The
randomness in per-insert cost comes entirely from which levels are
writes and which are reads; if write levels and read levels cost the
same ($c_L = c_R$), the per-insert cost would be deterministic at
$c_0 + d\,c_L$ regardless of leaf index. The cost variance is thus a
direct consequence of the cost asymmetry between level types.

Second, the coefficient of variation
\[
  \mathrm{CV}(g) \;=\; \frac{|c_L - c_R|\,\sqrt{d}/2}{c_0 + d\,(c_L + c_R)/2}
\]
approaches $|c_L - c_R| / \bigl(\sqrt{d}\,(c_L + c_R)\bigr)$ in the
large-$d$ limit and therefore \emph{decreases} as $\Theta(1/\sqrt{d})$,
even as the absolute standard deviation grows as $\Theta(\sqrt{d})$.
IMT cost becomes proportionally more predictable with depth while
becoming absolutely less so.

\subsection{PHDAG as a Degenerate Case}

The stochastic framework developed above for IMT covers PHDAG as
well, but as a trivial instance. For IMT, the warm-SSTORE count $S$
is a random variable: its value depends on the leaf-index bit
pattern, with mean $d/2 + 2$ and variance $d/4$
(Theorem~\ref{thm:moments}). For PHDAG, the analogous quantity is
\emph{deterministic}: every append performs exactly the same number
of storage operations (three cold SSTOREs and one warm SSTORE,
alongside the fixed two-or-three SLOADs of Lemma~\ref{lemma:nodep}),
independent of any input or registry state. As $S$ counts warm
SSTOREs by definition, the SLOADs sit outside the count $S$ but are
equally deterministic and equally independent of registry state.

A random variable that takes a single value with probability 1 is
called \emph{degenerate}; it has the same distribution as a
non-random constant. In this sense, $S_{\mathrm{PHDAG}} = 4$ almost
surely, with $\mathbb{E}[S_{\mathrm{PHDAG}}] = 4$ and
$\mathrm{Var}[S_{\mathrm{PHDAG}}] = 0$. PHDAG is not outside the
framework; it is the framework's trivial case, where all randomness
collapses.

This framing matters for what it predicts about empirical variance.
The framework attributes any non-zero variance in $g_{\mathrm{PHDAG}}$
to sources outside the structural cost model: calldata byte-level
variation (different counts of leading-zero bytes in the random
$\mathsf{arId}$), and other non-structural per-transaction effects.
The empirically observed $\sigma_{\mathrm{PHDAG}} \approx 6$ gas
(Section~\ref{section:results}), with measured costs falling on a
small number of discrete levels spaced $12$ gas apart, is consistent
with this prediction: the residual variance is small and traceable
to calldata composition rather than to any depth- or
registry-size-dependent quantity.

\subsection{The Crossover Depth}
\label{section:crossover}

Equating the corrected mean IMT cost (Equation~\ref{eq:meang}) with
the constant PHDAG cost (Equation~\ref{eq:phdag}) and solving for
$d$ gives the depth at which the two primitives have equal mean cost:
\begin{equation}
d^* \;=\; \frac{g_{\mathrm{PHDAG}} - c_0}{(c_L + c_R)/2}.
\label{eq:crossover}
\end{equation}
Using the calibrated constants $c_0 \approx 37{,}800$,
$c_L \approx 7{,}950$, $c_R \approx 2{,}750$ gas
(Section~\ref{section:results}) and the measured
$g_{\mathrm{PHDAG}} = 76{,}276$ gas,
\[
  d^* \;\approx\; \frac{76{,}276 - 37{,}800}{(7{,}950 + 2{,}750)/2}
        \;\approx\; \frac{38{,}476}{5{,}350}
        \;\approx\; 7.2.
\]
This uniform-index crossover is the apples-to-apples comparison:
the lifetime-average IMT cost (the mean over a full epoch of
$2^d$ sequential inserts, for which the popcount distribution is
exactly binomial; see Remark~\ref{rem:uniform}) against the
constant PHDAG cost. The depth sweep observes a slightly lower
crossover, at $d = 6$, and this is fully accounted for by the
sweep's index distribution rather than by any gap in the model.
The sweep populates each tree with the small leaf indices
$i = 1, 2, 3, \ldots$, whose low Hamming weight gives an
above-average write-level count: at $d = 5$ the sampled mean is
$W \approx 3.75$ against the uniform expectation $d/2 = 2.5$, and
at $d = 6$ it is $W \approx 4.6$ against $3.0$. Evaluating the
fitted cost model at the sweep's actual sampled indices reproduces
the measured per-depth means to within the residual scale of
Section~\ref{section:agreement} (for example $78{,}248$ predicted
versus $78{,}293$ measured at $d = 6$) and yields the same crossover
at $d = 6$. The uniform-index value $d^* \approx 7.2$ and the
as-sampled value $d = 6$ are therefore the same model under two
index distributions: the former is the lifetime average, the latter
the write-heavy low-index slice the sweep measured.

The practical conclusion is insensitive to this distinction. Whether
the crossover is taken as the lifetime-average $7.2$ or the
write-heavy as-sampled $6$, it lies far below the $d \geq 20$
of every production registry surveyed (Section~\ref{section:discussion}),
placing all such systems in the regime where PHDAG is strictly
cheaper.

\section{Methodology}
\label{section:methodology}

\subsection{Network and Tooling}

All experiments were conducted on the Base Sepolia testnet (chain
id $84532$), an Ethereum L2 network operating under the same gas
schedule as Ethereum mainnet. Base Sepolia provides a realistic EVM
environment at testnet cost. Contracts were compiled with
\texttt{solc 0.8.24} and deployed and invoked via \texttt{web3.py}
with local nonce tracking to avoid race conditions with
load-balanced RPC endpoints.

\subsection{Reference Implementations}

We implemented two minimal reference contracts, both available in
the reproducibility repository~\cite{AnchorResearch26}. The
\texttt{IncrementalMerkleTree} contract implements the
frontier-update pattern (as defined in
Section~\ref{section:imt-primitive})
with configurable depth at deployment. The \texttt{ParentHashDAG}
contract implements the minimal three-SSTORE append operation with
no application-specific metadata. Both contracts have been audited
for correctness against unit tests and deployed multiple times
across the experimental range.

\subsection{Experimental Design}

\textbf{Uniform depth sweep.} We deployed the IMT contract at each
depth $d \in \{1, 2, \ldots, 25\}$ and performed $3$ appends against
each instance. The first insert at each depth traverses cold
frontier slots and is discarded from the steady-state analysis;
the remaining inserts are averaged to estimate the per-insert cost
at that depth.

\textbf{Proportional depth sweep.} To obtain better statistical
power at production-relevant depths, we additionally deployed the
IMT contract at each depth and performed $\max(d, 2)$ appends per
depth, again discarding the first insert. At depth $25$ this
yields $24$ post-cold-cascade inserts per depth, affording
meaningful variance estimates.

\textbf{PHDAG run.} A single PHDAG contract was deployed and $200$
appends were performed to characterize the primitive's cost
distribution. Every append in this run was a \emph{root} append: the
parent reference was the zero value, so the parent-existence SLOAD
of Lemma~\ref{lemma:nodep} is short-circuited and each append
performs two cold SLOADs rather than three. The run therefore
measures the root case of the append; a non-root append (one
carrying a live parent reference) performs one additional cold
SLOAD, a fixed increment of $\gamma_{\mathrm{R}} \approx 2{,}100$
gas that does not depend on registry size or depth and so leaves
the constant-bound claim of Theorem~\ref{thm:phdag} intact. PHDAG
does not admit a depth parameter, so no sweep is required; its cost
is by construction invariant to registry size and to any notion of
tree depth (Theorem~\ref{thm:phdag}). Random leaf values throughout
were generated via \texttt{secrets.token\_bytes(32)} to avoid
pathological bit patterns. The first append against the freshly
deployed contract pays a one-time cold-initialization cost and is
excluded from the steady-state statistics below, exactly as the
first insert at each IMT depth is excluded from the sweep.

\subsection{Measurement}

Gas consumption per insert was obtained from the \texttt{gasUsed}
field of the post-mining transaction receipt. Receipts were
collected sequentially with local nonce tracking between
transactions to preclude concurrent execution. All raw data is
published alongside this manuscript in a reproducibility
repository~\cite{AnchorResearch26}.

\section{Results}
\label{section:results}

The PHDAG primitive exhibits mean gas cost
$\bar{g}_{\mathrm{PHDAG}} = 76{,}276$ gas with standard deviation
$\sigma_{\mathrm{PHDAG}} \approx 6$ gas across the steady-state
portion of the $200$-append root run (the $199$ appends remaining
after the first-insert cold-initialization outlier is excluded; see
Section~\ref{section:methodology}). The measured costs take only a
small number of discrete values spaced $12$ gas apart, tracking the
number of leading-zero bytes in the random $\mathsf{arId}$ calldata;
this $12$-gas quantum is the spacing between adjacent levels, not
the standard deviation, and the $199$ steady-state appends span a
range of just $36$ gas in total. The excluded first append costs
$93{,}379$ gas, roughly $17{,}100$ above the steady-state cluster, a
one-time cold-initialization cost on the freshly deployed contract;
including it raises the run mean to $76{,}361$ gas, and it is set
aside on the same basis as the IMT cold-cascade first insert.
The steady-state total sits below the nominal structural
expression of Equation~\ref{eq:phdag} for two reasons, both
structural rather than incidental. First, the benchmarked appends
are root appends, so the write to \texttt{parents[arId]} stores the
zero parent reference to an already-zero slot: a zero-store in place
of one of the three cold sets, removing roughly one $\gamma$ from the
nominal storage total (Lemma~\ref{lemma:nodep}). Second, the rounded
schedule values $\gamma \approx 22{,}100$ and
$\gamma_{\mathrm{R}} \approx 2{,}100$ are upper nominal figures, and
the value-transition pricing of EIP-2200/2929/3529 for the specific
cold/warm and zero/nonzero cases involved is somewhat lower. The
measured $76{,}276$ is therefore the root-case constant
($2\gamma_{\mathrm{R}}$ of SLOADs, two cold sets plus one zero-store,
and the warm counter write), not the non-root nominal sum; a non-root
append adds one cold SLOAD and turns the zero-store back into a cold
set, a fixed increment that leaves the structural claim unchanged.
What the formal analysis predicts, and what the measurement confirms,
is that the per-append cost is a fixed constant, independent of
registry size and of any notion of tree depth
(Theorem~\ref{thm:phdag}). The exact constant is a matter of which
append case is measured and of the schedule's effective per-operation
prices, not of the structure, and is recovered directly from the
measurement.

\subsection{Proportional Depth Sweep}

Figure~\ref{fig:scaling} displays the results of the proportional
depth sweep. Across the full range $d \in \{1, \ldots, 25\}$, the
per-insert cost of IMT grows approximately linearly with $d$. The
PHDAG primitive remains at $76{,}276$ gas across all
depths, consistent with the depth-invariance established by
Theorem~\ref{thm:phdag} and visualized as a horizontal reference
line in the figure. The per-insert cost also exhibits the
leaf-index dependence predicted by Lemma~\ref{lemma:writes}: inserts
at indices $2^k$ (longest write cascades, where many popcount bits
are zero) are the most expensive, and those at $2^k - 1$ (all low
bits one) the cheapest.

\begin{figure}[h]
\centering
\includegraphics[width=0.95\linewidth]{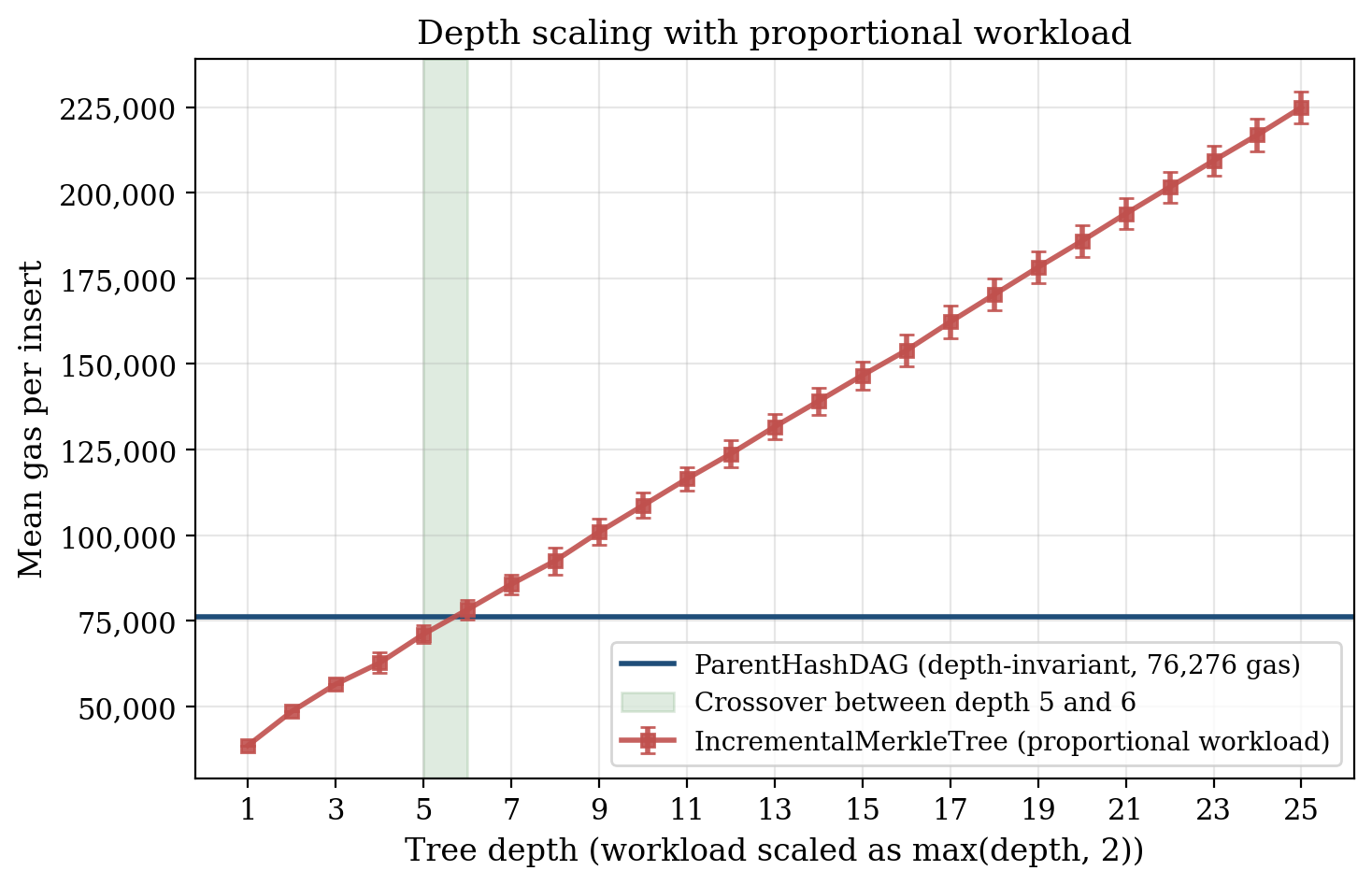}
\caption{Per-insert gas cost of PHDAG and IMT primitives as a
function of tree depth. PHDAG (blue) is depth-invariant at
$76{,}276$ gas. IMT (red) grows approximately linearly with depth,
with per-insert variance (error bars) also growing in absolute
terms. The crossover between primitives under the sweep's workload
falls at depth $6$ (green band), slightly below the lifetime-average
$d^* \approx 7.2$ (Section~\ref{section:crossover}). Each IMT point
aggregates $\max(d, 2)$ appends per depth with the first insert
excluded to remove cold-cascade contamination.}
\label{fig:scaling}
\end{figure}

The empirical crossover depth, the value of $d$ at which
$\bar{g}_{\mathrm{IMT}}$ first exceeds $\bar{g}_{\mathrm{PHDAG}}$,
is located at $d = 6$ in the sweep
($\bar{g}_{\mathrm{IMT}} \approx 71{,}000$ at $d = 5$, rising to
$\approx 78{,}000$ at $d = 6$). As discussed in
Section~\ref{section:crossover}, this as-sampled crossover sits
slightly below the lifetime-average value $d^* \approx 7.2$ because
the sweep samples write-heavy low indices; both lie far below
production depths. For all depths $d \geq 6$ the PHDAG primitive is
strictly cheaper per append under the sweep's workload.

\subsection{Variance Scaling}

The error bars in Figure~\ref{fig:scaling} encode the empirical
per-insert standard deviation at each depth. IMT variance grows
monotonically with depth, consistent with the $\Theta(\sqrt{d})$
scaling of Theorem~\ref{thm:cost}. The absolute magnitude of the
sweep's per-depth standard deviation is, however, smaller than the
uniform-index prediction $|c_L - c_R|\sqrt{d}/2$: at $d = 10$ the
uniform-index formula gives $\approx 8{,}200$ gas, while the sweep's
sample standard deviation is $\approx 3{,}500$ gas. This is the same
sampling effect that governs the mean (Section~\ref{section:agreement}):
the sweep measures only small leaf indices, which span a narrow band
of Hamming weights and therefore a narrow band of write-level counts
$W$, under-representing the full popcount range that drives the
uniform-index variance. The sweep confirms the \emph{sign and
functional form} of the variance scaling (monotone growth in $d$)
but not its uniform-index magnitude, which would require sampling
the full leaf-index space at each depth.

PHDAG variance, by contrast, is invariant to depth at
$\sigma_{\mathrm{PHDAG}} \approx 6$ gas across all measurements,
consistent with the zero-variance prediction of
Theorem~\ref{thm:phdag} modulated only by the discrete $12$-gas
calldata quantum described above.

\subsection{Agreement with the Stochastic Model}
\label{section:agreement}

We estimate the per-level cost constants of Theorem~\ref{thm:cost}
directly from the proportional depth sweep (Section~\ref{section:methodology}):
the $326$ measured inserts across depths $1$--$25$, less the $25$
first-inserts excluded for the cold-allocation cascade, giving
$n = 301$. For each measured insert we
compute the write-level count $W = d - \mathrm{pop}(i)$
(Lemma~\ref{lemma:writes}) from its depth and leaf index, and fit
the cost model
\[
  g \;=\; c_0 + c_L\,W + c_R\,(d - W)
\]
by ordinary least squares (OLS). Because $W$ and $(d - W)$ sum to
$d$, the model is equivalently written in the cleanly identified
form $g = c_0 + c_R\,d + (c_L - c_R)\,W$, separating an intercept,
a per-depth cost $c_R$, and a per-write premium $c_L - c_R$; we
report all three coefficients together with their standard errors.
The first insert at each depth, which pays the one-time
zero-to-nonzero cold-allocation cascade as the frontier slots are
populated for the first time, is handled either by exclusion or by
an additive indicator term; the two treatments yield coefficients
agreeing to within one gas, confirming that the cascade is
separable from the steady-state per-level costs.

The excl-first-insert fit ($n = 301$) gives
\[
  c_0 = 37{,}823 \ (\mathrm{SE}\ 23), \quad
  c_L = 7{,}952 \ (\mathrm{SE}\ 1.2), \quad
  c_R = 2{,}748 \ (\mathrm{SE}\ 7.9),
\]
with coefficient of determination $R^2 = 0.999994$ and root-mean-square
residual $113$ gas on cost values spanning $38{,}635$ to
$253{,}706$ gas. The per-write premium is $c_L - c_R = 5{,}204$
(SE $8.2$). The near-exact fit, with a structurally simple
three-parameter model, confirms both the level-type decomposition
of Theorem~\ref{thm:cost} and the constancy of the per-level costs
across the full depth range.

The calibrated constants give a uniform-index mean slope of
$(c_L + c_R)/2 \approx 5{,}350$ gas per depth level
(Equation~\ref{eq:meang}). The slope measured directly from the
sweep is higher, at $\approx 7{,}700$ gas per depth level, and the
difference is a sampling effect rather than a discrepancy in the
model. The sweep populates each tree with small leaf indices
$i = 1, 2, 3, \ldots$, whose low Hamming weight gives an
above-average write-level count $W = d - \mathrm{pop}(i)$; the
measured inserts are therefore write-heavy relative to the uniform
mean $\mathbb{E}[W] = d/2$, and the more expensive write levels
($c_L > c_R$) are over-represented. Reweighting the per-level
costs by the sweep's empirical write fraction reproduces the
measured slope; equivalently, evaluating the fitted model at the
sweep's actual $W$ values reproduces every per-depth mean to within
the residual scale above. The uniform-index slope $5{,}350$ and the
as-sampled slope $7{,}700$ are the same model under two index
distributions.

\section{Discussion}
\label{section:discussion}

\subsection{Regimes of Applicability}

The crossover at $d \approx 6$--$7$ (Section~\ref{section:crossover})
divides the design space into two regimes. The crossover is a narrow
band rather than a single depth (the as-sampled $6$ to the
lifetime-average $7.2$); the two-regime split below brackets this
band, and every production registry surveyed sits far into the deep
regime regardless of which edge is taken.

\textbf{Shallow regime ($d \lesssim 6$).} IMT is more gas-efficient
per append and additionally provides succinct $O(\log N)$ inclusion
proofs, a property PHDAG does not provide without additional
structure. For applications requiring membership proofs with
bounded capacity, IMT is the correct primitive. At the floor of
the IMT range, $d = 1$, an append writes at most one
\texttt{filledSubtrees} slot plus the unconditional \texttt{root}
and \texttt{nextIndex} slots: two to three SSTOREs total against
PHDAG's three cold SSTOREs plus a warm counter write, which
explains the cost gap at the shallowest depths.

\textbf{Deep regime ($d \gtrsim 7$).} PHDAG is strictly cheaper per
append, and additionally exhibits negligible per-insert variance.
This is the regime occupied by every production append-only
on-chain registry we have surveyed: Tornado Cash ($d = 20$),
Semaphore ($d = 20$--$32$), and the major
zero-knowledge rollup state commitments ($d = 32$--$40$). In this
regime the choice between primitives is governed by whether the
application requires succinct inclusion proofs.

\subsection{The Role of Inclusion Proofs}

The comparison in this paper addresses append cost only. IMT's
$O(\log N)$ inclusion proof is a valuable property not shared by
PHDAG: a Merkle inclusion proof~\cite{Laurie13} is a sequence of
$d$ sibling hashes that can be verified by a contract or off-chain
verifier without reference to the remainder of the tree. ZK rollups and
privacy protocols exploit this property heavily, and for these
applications the depth-dependent append cost is accepted as a
necessary consequence.

Provenance registries, by contrast, typically require different
queries: parent-chain traversal, descendant enumeration, and
direct existence checks by unique identifier. The $O(|V|)$
off-chain reconstruction guarantee of Proposition~\ref{prop:reconstruct}
supports all of these queries natively in linear time in the
registry size, with no additional on-chain structure required.
The architectural match between query pattern and primitive is
the determining factor in primitive selection.

\subsection{Variance and Operational Predictability}

For production operators providing append services on behalf of
end users, as is typical in operator-gated
registries~\cite{Moore26Trustless}, per-insert cost
predictability is an operational requirement distinct from average
cost. An operator quoting a fixed fee per append must either
over-charge (to cover worst-case cost) or accept margin
compression on the expensive tail of the distribution. PHDAG's
near-zero variance eliminates this consideration entirely. IMT's
variance, while small in relative terms (CV decreases with depth),
grows in absolute terms and at depth $25$ is a non-trivial
fraction of the per-insert cost.

The practical implication is that IMT-based registries operating
at production depth must budget for the upper tail of the
per-insert distribution, not the mean. For an operator submitting
$10^5$ appends per day at depth $25$, a plausible volume for
a large ZK rollup or privacy protocol, the cumulative variance
of the daily gas budget is non-trivial. PHDAG-based registries do
not incur this budgeting uncertainty.

\subsection{Trustless Reconstruction as a Robustness Property}

Proposition~\ref{prop:reconstruct} formalizes a property that, in
practice, determines the \emph{indestructibility} of a provenance
registry. If every off-chain indexing service associated with a
PHDAG-based registry were to disappear simultaneously, a complete
reconstruction of the registry could be performed by any party
with read access to the blockchain in time linear in the registry
size. No trust in any single party is required; no off-chain
infrastructure is required; no coordination among users is
required. The registry survives as long as the underlying
blockchain survives.

This robustness derives from the log-canonical architecture of
Section~\ref{section:tiers}: event logs are attested through the
receipts root of every block header, so their integrity is
verifiable from chain headers alone, independent of any contract's
current storage state or any off-chain index. The insertion event
in PHDAG carries the complete metadata required for reconstruction,
which is the architectural property that makes the indestructibility
claim rigorous rather than rhetorical. For PHDAG this property is
architectural: it follows from the log-canonical design by
construction (Proposition~\ref{prop:reconstruct}), independent of
any implementation choice. For IMT it is at most incidental. An
IMT-based registry can recover the same guarantee only by emitting
sufficient payload in its insertion events, that is, by replicating
in its auxiliary log the very state its storage-canonical design
treats the Merkle root as carrying. The property is therefore native
to one primitive and bolted on to the other, and the difference
traces directly to where each places its canonical history
(Section~\ref{section:tiers}).

\subsection{Implications for Standardization}

The PHDAG primitive's simplicity and depth-invariance make it a
natural candidate for standardization as an Ethereum Improvement
Proposal (EIP) defining a uniform interface for append-only
provenance registries. A standardized PHDAG interface would enable
interoperable provenance infrastructure across the ecosystem
without requiring each application to reimplement storage layout,
event schema, or reconstruction logic. The constant-cost guarantee
of Theorem~\ref{thm:phdag} and the linear-time reconstruction
guarantee of Proposition~\ref{prop:reconstruct} are both
implementation-agnostic and therefore suitable for inclusion in
such a standard. We leave the precise specification of the
proposed EIP to future work.

\subsection{Provenance Traces as a Coordination Medium}
\label{section:coordination-discussion}

The preceding subsections frame the choice between PHDAG and IMT
in terms of append cost, variance, and inclusion proofs. A
complementary framing is available, and it is the one that most
directly explains why PHDAG suffices for provenance applications
despite lacking succinct membership proofs. PHDAG is a
\emph{coordination medium} in the sense introduced in
Section~\ref{section:coordination-medium}: an append-only ledger
whose insertion events are read by an ecosystem of autonomous
agents (indexers, derivative-work registries, compliance tools,
AI provenance verifiers) which then act on those traces. In the
pattern catalogue of~\cite{Paredes26}, this is the
\emph{Event-Signal} pattern at its strongest, natural for a
log-canonical primitive whose insertion events \emph{are} the
canonical history rather than auxiliary notifications. Each PHDAG
append emits a single event carrying the complete metadata
required for reconstruction
(Proposition~\ref{prop:reconstruct}); the event is the trace, the
agents are the consumers of provenance, and the activation
predicate is application-specific (a new descendant of a known
artifact, a manifest matching a query, a tree whose root has been
flagged).

This reframes the inclusion-proof question. ZK rollups and privacy
protocols require succinct membership proofs because their agents
are on-chain contracts that must verify membership within a single
transaction's gas budget. Provenance registries are consumed by a
different agent population: off-chain readers whose activation
predicates run over the full event stream, not within an in-circuit
verification. The $O(\log N)$ inclusion proof that motivates IMT is
the wrong primitive for this consumer profile; the right primitive
is a complete, authenticated, linear-time reconstructible event
log, which is precisely what PHDAG provides by construction. The
architectural advantage is not that PHDAG could be augmented with
succinct proofs; it is that the consumer population does not need
them.

A second consequence concerns the persistence of traces. Paredes
Garc\'ia~\cite{Paredes26} observes that biological stigmergic media
decay (pheromones evaporate), while ledger state does not; that
work treats the asymmetry as a design consideration with both costs
(stale traces persist, polluting the predicate space) and benefits
(the medium is permanently legible). Provenance registries occupy
the extreme end of this spectrum: the indestructibility of the
trace is itself the value proposition. A retracted or expired
anchor would undermine the core claim of an immutable provenance
record. PHDAG's design, append-only with no decay mechanism and
reconstructible from public event logs in $O(|V|)$ time, aligns
with this requirement by construction. The properties that would
be liabilities in a task-board contract (no expiry, no garbage
collection, unbounded state growth) are exactly what a provenance
medium requires.

\section{Related Work}
\label{section:related}

\textbf{Incremental Merkle trees.} The \texttt{MerkleTreeWithHistory}
contract of Tornado Cash~\cite{TornadoCash19} established the
canonical production IMT pattern. Semaphore~\cite{Semaphore20}
adapted the pattern for identity commitment sets. The major
zero-knowledge rollups (zkSync~\cite{ZkSync22},
Scroll~\cite{Scroll23}, Linea~\cite{Linea23}) employ sparse Merkle
tree variants at depths $32$--$40$ as their state commitment
mechanism. Our analysis applies to the frontier-update pattern
(as defined in Section~\ref{section:imt-primitive})
common to these systems; sparse Merkle trees with account-keyed
access patterns exhibit additional cost variation not modeled
here.

\textbf{Parent-hash provenance.} The underlying parent-reference append
pattern, in which each appended element carries a pointer to a prior
element and is registered in unique-keyed storage, is most fully
developed as the data-structural core of the AnchorRegistry
provenance registry of Moore~\cite{Moore26Trustless}. That work
embeds the pattern within a dual-layer cryptographic commitment
scheme addressing operator trust, and asserts an $O(1)$ append
complexity at a high level in its implementation section. The
present paper isolates the pattern from its application context,
names it the parent-hash directed acyclic graph (PHDAG) primitive,
formalizes its complexity with explicit constants, and places it
in direct comparison with IMT within a unified stochastic cost
analysis.

\textbf{Gas cost analysis.} Prior work on EVM gas cost has focused
primarily on deterministic cost modeling for
optimization~\cite{Albert18,Grech18} and on the economic design of
gas markets~\cite{Rou20,But19}. The stochastic treatment of
per-insert cost as a random variable parameterized by usage
patterns, to our knowledge, has not been presented in prior work.

\textbf{Ledger-state coordination.} Paredes
Garc\'ia~\cite{Paredes26} develops ledger-state stigmergy as a
formal framework for indirect coordination of autonomous agents
through shared ledger state, introducing a state-transition
formalism and an applied pattern catalogue (State-Flag,
Event-Signal, Threshold-Trigger, with a Commit-Reveal sequencing
overlay). That work is medium-general and uses a task-board as its
worked example; the present paper applies its vocabulary in
Sections~\ref{section:coordination-medium}
and~\ref{section:coordination-discussion} to characterize PHDAG as
an Event-Signal coordination medium for provenance traces.

\textbf{Certificate Transparency.} On-chain CT-style
logs~\cite{Laurie13} employ similar append-only patterns but
typically operate on-chain only as a settlement layer, with the
bulk of the log maintained off-chain. Our analysis applies
directly to any variant that performs per-append on-chain work.

\section{Conclusion}
\label{section:conclusion}

We have presented both a formal and a stochastic cost analysis of
the two dominant append-only on-chain registry primitives: the
incremental Merkle tree and the parent-hash directed acyclic graph.
The formal analysis establishes that PHDAG append is $O(1)$ in
both global registry size and any notion of tree depth, with an
explicit constant bound derived from the EVM gas schedule. The
stochastic analysis yields closed-form expressions for the mean
and variance of IMT per-insert cost, with the variance growing as
$\Theta(\sqrt{d})$ in absolute terms and decreasing as
$\Theta(1/\sqrt{d})$ in relative terms.

We validated both analyses empirically on Base Sepolia across
tree depths $1$--$25$. The PHDAG primitive was observed to cost
$76{,}276$ gas per append with a standard deviation of
approximately $6$ gas, consistent with the constant bound
predicted by the formal analysis. The IMT primitive was observed
to grow linearly with depth, with the crossover between primitives
at $d^* \approx 7$ for a uniform leaf index and at $d = 6$ under the
depth sweep's workload. Every production-deployed append-only
registry surveyed operates at a depth above this crossover, placing
all such systems in the regime where PHDAG is strictly cheaper and
where per-append cost is fully predictable.

A separate structural property, $O(|V|)$ trustless
reconstruction from public event logs alone, ensures that
PHDAG-based registries survive the failure of any off-chain
infrastructure. Constant-cost append, depth-invariant variance,
and trustless reconstruction together make PHDAG the preferred
primitive for append-only registries operating at production
depth, where succinct membership proofs are not the dominant
query, the regime every production registry surveyed occupies,
and a strong candidate for standardization as a uniform registry
interface for the Ethereum ecosystem.

The stochastic framing suggests several directions for future
work. The depth sweep measured a write-heavy sub-sample of the
leaf-index space; a sweep over uniformly sampled leaf indices at
each depth would directly confirm the uniform-index variance
magnitude predicted by Theorem~\ref{thm:cost}, which the present
experiments confirm only in functional form. Extension of the
analysis to sparse Merkle tree variants used in zero-knowledge
rollups would characterize the cost profile of those systems' state
commitment mechanisms.
Finally, the standardization of the PHDAG primitive as a uniform
interface for append-only provenance infrastructure, through an
Ethereum Improvement Proposal, would enable interoperable
provenance systems across the ecosystem.

\section*{Acknowledgments}

The authors thank the Base team for operating a reliable L2
testnet on which this work was conducted. The authors used AI
writing assistance (Anthropic Claude) in drafting this manuscript.
All theoretical contributions, experimental designs, and empirical
analyses are the authors' own. The reference implementations, raw
experimental data, and analysis scripts are available
at~\cite{AnchorResearch26}.

\end{document}